\documentclass[11pt]{amsart}
\usepackage{amsmath}
\usepackage{mathrsfs}
\usepackage{amssymb}
\usepackage{amsfonts}
\usepackage{hyperref}
\usepackage[utf8]{inputenc}
\usepackage{cite,amsthm,xypic,cases,graphicx,color}
\newtheorem{theorem}{Theorem}
\newtheorem{lemma}[theorem]{Lemma}

\newtheorem{corollary}[theorem]{Corollary}
\theoremstyle{remark}

\newcommand{\C}{{\mathbb C}}
\newcommand{\N}{{\mathbb N}}
\newcommand{\R}{{\mathbb R}}
\newcommand{\Z}{{\mathbb Z}}
\newcommand{\Sph}{{\mathbb S}}

\newcommand{\Ree}{{\mbox{Re}}}

\newcommand{\bnri}{{N\rightarrow\infty}}

\begin{document}
\title{Next order energy asymptotics for Riesz potentials on flat tori}
\date{}

\author[D.P. Hardin]{Douglas P. Hardin}
\author[E.B. Saff]{Edward B. Saff}
\author[B.Z. Simanek]{Brian Z. Simanek}
\author[Y. Su]{Yujian Su}

\date{\today}
\thanks{This research was supported, in part,
by the U. S. National Science Foundation under the grant DMS-1412428 and DMS-1516400.
}

\begin{abstract}
Let $\Lambda$ be a lattice in $\R^d$ with positive co-volume.  Among  $\Lambda$-periodic $N$-point configurations, we consider the  minimal renormalized Riesz $s$-energy  $\mathcal{E}_{s,\Lambda}(N)$.   
While the dominant term in the asymptotic expansion of $\mathcal{E}_{s,\Lambda}(N)$ as $N$ goes to infinity  in the long range case that $0<s<d$ (or $s=\log$) can be obtained from classical potential theory, the next order term(s) require a different approach.     Here we derive the form of  the next order term or terms, namely for $s>0$ they are of the form 
$C_{s,d}|\Lambda|^{-s/d}N^{1+s/d}$ and $-\frac{2}{d}N\log N+\left(C_{\log,d}-2\zeta'_{\Lambda}(0)\right)N$ where we show that the constant $C_{s,d}$ is independent of the lattice $\Lambda$. 
\end{abstract}

\address{Center for Constructive Approximation\\
         Department of Mathematics\\
         Vanderbilt University\\
         1326 Stevenson Center\\
         Nashville, TN, 37240\\
         USA}
\email{doug.hardin@vanderbilt.edu}
\email{edward.b.saff@vanderbilt.edu}
\email{BrianSimanek@baylor.edu}
\email{yujian.su@vanderbilt.edu}



\maketitle

\section{Preliminaries}

Let  $A= [v_1,\ldots, v_d]$ be a $d\times d$ nonsingular matrix   with $j$-th column $v_j$ and let $\Lambda=\Lambda_A:=A\Z^d$ denote the lattice generated by $A$.  The set
\[
\Omega=\Omega_\Lambda:=\left\{w:w=\sum_{j=1}^d\alpha_jv_j,\, \alpha_j\in[0,1),\, j=1,2,\ldots,d\right\}.
\]
 is a   fundamental domain of the quotient space $\R^d/\Lambda$; i.e., the collection of sets $\{\Omega +v\colon v\in \Lambda \}$ tiles $\R^d$.  The volume of
 $\Omega_\Lambda$, denoted by $|\Lambda|$, equals $|\det A|$ and  is called the co-volume of $\Lambda$ (in fact, any measurable fundamental domain of $\Lambda$ has the same volume).  We will let  $\Lambda^*$ denote the the {\em dual lattice} of $\Lambda$ which is the lattice generated by $(A^T)^{-1}$.

 For an interaction potential $F:\R^d\to \R\cup\{+\infty\}$, we consider the {\em $F$-energy} of an $N$-tuple  $\omega_N=(x_1,\ldots, x_N)\in (\R^d)^N$  \begin{equation}
 E_F(\omega_N):=\sum_{k=1}^N\sum_{\substack{j=1\\j\neq k}}^N F(x_k-x_j),
 \end{equation}
 and for a subset $A\subset \R^d$, we consider the {\em $N$-point minimal $F$-energy}
 \begin{equation}
 \mathcal{E}_F(A,N):=\inf_{\omega_N\in A^N} E_F(\omega_N).
 \end{equation}
In this paper we are mostly concerned with  $\Lambda$-periodic potentials $F$, that is, $F(x+v)=F(x)$ for all $v\in \Lambda$.  For such an $F$,
the energy  $E_{ F}(\omega_N)=E_{ F}(x_1,\ldots, x_N)$  is $\Lambda$-periodic in each component $x_k$ and so, without loss of generality,  we may assume that $\omega_N\in (\Omega_\Lambda)^N$; i.e., $ \mathcal{E}_{ F}(\R^d,N)=\mathcal{E}_{ F}(\Omega_\Lambda,N) $.   Specifically, in this paper we consider periodized Riesz potentials and periodized logarithmic potentials and $A=\R^d$ (or, equivalently $A=\Omega_\Lambda$) as we next describe.

 For $s>d$, we consider the periodic potential generated by the Riesz $s$-potential $f_s(x)=|x|^{-s}$ as follows
\begin{equation}\label{Fscp}
\zeta_{\Lambda}(s;x):=\sum_{v\in \Lambda}\frac{1}{|x+v|^s},\qquad s>d, x\in \R^d,
\end{equation}
 which is finite for $x\not \in \Lambda$ and equals $+\infty$ when $x\in \Lambda$. Then $ \zeta_{\Lambda}(s;x-y)$ can be considered to be the energy required to place a unit charge at location $x\in \R^d$  in the presence
 of unit charges placed at $y+\Lambda=\{y+v\colon v\in \Lambda\}$ with charges interacting through the Riesz $s$-potential.
   For $s\le d$, the sum on the right side of \eqref{Fscp} is infinite for all $x\in \R^d$.  In \cite{hardin2014periodic},  $\Lambda$-periodic energy problems for a class of long range potentials are considered and it is shown that for the case of the Riesz potential  with $s\le d$, the appropriate energy problem can be obtained through analytic continuation.  Specifically,    it follows from
   Theorems 1.1 and 3.1 of \cite{hardin2014periodic} that the potential $F_{s,\Lambda} (x)$ defined by
\begin{equation}\label{Fsdef} F_{s,\Lambda} (x):=\sum_{v\in\Lambda}\int_1^\infty e^{-|x+v|^2t}\frac{t^{\frac{s}{2}-1}}{\Gamma(\frac{s}{2})}\mathrm{d}t+\frac{1}{|\Lambda|}\sum_{w\in\Lambda^*\setminus\{0\}}e^{2\pi i w\cdot x}\int_0^1 \frac{\pi^{\frac{d}{2}}}{t^\frac{d}{2}}e^{-\frac{\pi^2|w|^2}{t}}\frac{t^{\frac{s}{2}-1}}{\Gamma(\frac{s}{2})}\mathrm{d}t,
\end{equation}
  is, for  fixed  $x\in\R^d\setminus \Lambda$,  an entire function of $s$, and
\begin{equation}\label{FsLam}
  F_{s,\Lambda} (x)=\zeta_{\Lambda}(s;x)+\frac{2\pi^{\frac{d}{2}}|\Lambda|^{-1}}{\Gamma(\frac{s}{2})(d-s)}, \qquad s>d,\end{equation}
showing that \eqref{Fsdef} provides an analytic continuation of $\zeta_{\Lambda}(\cdot;x)$  to $\C\setminus \{d\}$ (note that $\zeta_{\Lambda}(s;x)$ has a simple pole at $s=d$ for $x\not\in\Lambda$).   We refer to (the analytically extended) $\zeta_{\Lambda}(s;x)$ as  the {\em Epstein Hurwitz zeta function for the lattice $\Lambda$}.  We shall also need the {\em Epstein zeta function} defined for $s>d$ by
\begin{equation}
 \zeta_{\Lambda}(s):=\sum_{v\in \Lambda\setminus\{0\}}\frac{1}{|v|^{s}},
\end{equation}
and continued analytically as above for $s\in \C\setminus \{d\}$. We remark that \eqref{Fsdef} is derived from
the formula
\begin{equation}\label{eq2}
   r^{-s}=\int_0^\infty e^{-tr^2}\frac{t^{\frac{s}{2}-1}}{\Gamma(\frac{s}{2})}\mathrm{d}t, \qquad (r>0),
\end{equation}
   together with the Poisson Summation Formula.

In \cite{hardin2014periodic}, analytic continuation and periodized Riesz potentials are connected through the use of {\em convergence factors}; i.e., a parametrized family of functions $g_a:\R^d\to [0,\infty)$ such that
\begin{itemize}
\item[(a)] for $a>0$, $f_s(x)g_a(x) $ decays sufficiently rapidly as $|x|\to \infty$ so that $$F_{s,a,\Lambda} (x):=\sum_{v\in \Lambda}f_s(x+v)g_a(x+v)$$ converges to a finite value for all $x\not \in \Lambda$, and
\item[(b)]  $\lim_{a\to 0^+}g_a(x)=1$ for all $x\in \R^d\setminus \{0\}$.
\end{itemize}
For example, the  family of Gaussians $g_a(x)=e^{-a|x|^2}$ is a convergence factor for Riesz potentials.  In \cite{hardin2014periodic}, it is shown that for a large class of convergence factors  $\{g_a\}_{a>0}$ (including the Gaussian convergence family) one may choose $C_a$ (depending on the convergence factor $\{g_a\}_{a>0}$) such that
\begin{equation}\label{convfac}F_{s,\Lambda} (x)=\lim_{a\to 0^+} \left(F_{s,a,\Lambda} (x)-C_a\right).\end{equation}
Then, for $a>0$,   $F_{s,a,\Lambda} (x-y)$   represents the energy required to place a unit charge at location $x$  in the presence
 of unit charges placed at $y+\Lambda=\{y+v\colon v\in \Lambda\}$ with charges interacting through the potential $f_s(x)g_a(x)$.
 This leads us to  consider, for $s>0$, the \textit{periodic Riesz $s$-energy of $\omega_N$ associated with the lattice $\Lambda$} defined by
\begin{align}\label{e1def}
E_{s,\Lambda}( \omega_N):=\sum_{{1\leq k,j\leq N}\atop{k\neq j}}F_{s,\Lambda} (x_k-x_j),
\end{align}
as well as the {\em minimal $N$-point periodic Riesz $s$-energy}
\begin{equation}
\mathcal{E}_{s,\Lambda} ( N):=\mathcal{E}_{F_{s,\Lambda} }(\R^d;N)=\inf_{\omega_N\in (\R^d)^N} E_{s,\Lambda} ( \omega_N).
\end{equation}

We shall also consider the \textit{periodic logarithmic potential associated with $\Lambda$}   generated  from the logarithmic potential $F_{\log}(x):=\log(1/|x|)$  using convergence factors as above and resulting in the definition
 \begin{equation}\label{Flogdef}
 F_{\log,\Lambda}(x)  :=\sum_{v\in\Lambda}\int_1^\infty e^{-|x+v|^2t}\frac{\mathrm{d}t}{t}+\frac{1}{|\Lambda|}\sum_{w\in\Lambda^*\setminus\{0\}}e^{2\pi i w\cdot x}\int_0^1 \frac{\pi^{\frac{d}{2}}}{t^{\frac{d}{2}}}e^{-\frac{\pi^2|w|^2}{t}}\frac{\mathrm{d}t}{t}.
 \end{equation}
 Comparing \eqref{Flogdef} and \eqref{Fsdef}, it is not difficult to obtain (cf. \cite{hardin2014periodic}) the relations
 \begin{equation}\label{Flog2}
 \begin{split} F_{\log,\Lambda} (x) &=\lim_{s\to 0}\Gamma(\frac{s}{2})F_{s,\Lambda} (x)=2\left(\frac{d}{ds}F_{s,\Lambda} (x)\right)\bigg|_{s=0}= 2\zeta_\Lambda'(0;x) +\frac{2\pi^{d/2}}{d}|\Lambda |^{-1},
 \end{split}\end{equation}
where the prime  denotes  differentiation with respect to the variable $s$.  We then define the {\em periodic logarithmic energy of $\omega_N=(x_1,\ldots, x_N)$},
\begin{equation}\label{elogdef}
E_{\log,\Lambda} ( \omega_N):=\sum_{\substack{1\leq k,j\leq N \\ k\neq j}}F_{\log,\Lambda} (x_j-x_k),
\end{equation}
 and also the \textit{$N$-point minimal periodic logarithmic energy for $\Lambda$},
\begin{equation}\label{minendeflog}
\mathcal{E}_{\log,\Lambda} (N):=\inf_{\omega_N\in (\R^d)^N}E_{\log,\Lambda} ( \omega_N).\end{equation}

\medskip

For $0<s<d$, the kernel $K_{s,\Lambda} (x,y):=F_{s,\Lambda} (x-y)$ is  positive definite and integrable on $\Omega_\Lambda\times \Omega_\Lambda$ and so there is a unique probability measure  $\mu_s$ (called the {\em Riesz $s$-equilibrium measure}) that minimizes  the continuous Riesz $s$-energy
$$I_{s,\Lambda}(\mu):=\iint_{\Omega_\Lambda \times \Omega_\Lambda} K_{s,\Lambda} (x,y)\mathrm{d}\mu(x)\mathrm{d}\mu(y)$$  over all Borel
probability measures $\mu$ on $\Omega_\Lambda$.   From the periodicity of $F_{s,\Lambda} $ and the uniqueness of the equilibrium measure, it follows that  $\mu_s=\lambda_d$ where $\lambda_d$ denotes  Lebesgue measure restricted to $\Omega_\Lambda$ and normalized so that $\lambda_d(\Omega_\Lambda)=1$; i.e., $\lambda_d$ is the  normalized Haar measure for $\Omega_\Lambda=\R^d/ \Lambda$. The  periodic logarithmic kernel $K_{\log,\Lambda} (x,y):=F_{\log,\Lambda} (x-y)$ is conditionally positive definite and integrable and it similarly follows that $\lambda_d$ is the unique equilibrium measure  minimizing the periodic logarithmic energy
$$I_{\log,\Lambda}(\mu):=\iint_{\Omega_\Lambda \times \Omega_\Lambda} K_{\log,\Lambda} (x,y)\mathrm{d}\mu(x)\mathrm{d}\mu(y)$$  over all Borel
probability measures $\mu$ on $\Omega_\Lambda$.

 lt is not difficult to verify (cf. \cite{hardin2014periodic}) that
\begin{equation}\int_{\Omega_\Lambda}\zeta_{\Lambda}(s;x)\ {\rm d}\lambda_d(x)=0 ,\qquad 0<s<d,\end{equation}
and \begin{equation}\int_{\Omega_\Lambda}\zeta_{\Lambda}'(0;x)\ {\rm d}\lambda_d(x)=0,\end{equation}
from which we obtain
\begin{equation}
I_{s,\Lambda}(\lambda_d)= \frac{2\pi^{\frac{d}{2}}|\Lambda|^{-1}}{\Gamma(\frac{s}{2})(d-s)},  \qquad 0<s<d,
 \end{equation}
 and
 \begin{equation}
I_{\log,\Lambda}(\lambda_d)= \frac{2\pi^{d/2}}{d}|\Lambda|^{-1}.
 \end{equation}
It then follows (cf. \cite{landkof1972foundations}) that
 \begin{equation}
 \lim_{N\rightarrow \infty} \frac{\mathcal{E}_{s,\Lambda} (N)}{N^2}=\frac{2\pi^{\frac{d}{2}}|\Lambda|^{-1}}{\Gamma(\frac{s}{2})(d-s)}, \qquad 0<s<d,
 \end{equation}
 and
 \begin{equation}
 \lim_{N\rightarrow \infty} \frac{\mathcal{E}_{\log,\Lambda} (N)}{N^2}=\frac{2\pi^{d/2}}{d}|\Lambda|^{-1}.
 \end{equation}

\section{Main Results}

Our main result is the following asymptotic expansion of the periodic Riesz and logarithmic minimal energy as $N\to \infty$.
\begin{theorem} \label{thm2} Let $\Lambda$ be a lattice in $\R^d$ with co-volume $|\Lambda|>0$.  Then, as $N\to\infty$,
  \begin{align}\label{Csd0}
  \mathcal{E}_{s,\Lambda} (N)&=\frac{2\pi^{\frac{d}{2}}|\Lambda|^{-1}}{\Gamma(\frac{s}{2})(d-s)}N^2+C_{s,d}|\Lambda|^{-s/d}N^{1+\frac{s}{d}}+o(N^{1+\frac{s}{d}}), \qquad 0<s<d,\\
  \mathcal{E}_{\log,\Lambda} (N)&=\frac{2\pi^{\frac{d}{2}}}{d}|\Lambda|^{-1}N(N-1)-\frac{2}{d}N\log N+\left(C_{\log,d}-2\zeta'_{\Lambda}(0)\right)N+o(N).\label{Clog0}
  \end{align}
where $C_{\log,d}$ and $C_{s,d}$ are constants  independent of $\Lambda$.
\end{theorem}

  Petrache and Serfaty establish in \cite{petrache2014next} a  result closely related to \eqref{Csd0} for  point configurations interacting through a Riesz $s$ potential and confined by an external field for values of the Riesz parameter $d-2\le s\le d$ and  Sandier and Serfaty prove in \cite{sandier20122d} a result closely related to \eqref{Clog0} for the case that $s=\log$ and $d=2$.  \\

 For comparison, when $s\geq d$ it is known that the leading order term of $\mathcal{E}_{s,\Lambda} (N)$ is the same as that of $\mathcal{E}_{s}( \Omega_\Lambda, N):=\mathcal{E}_{f_s}(\Omega_\Lambda, N)$.

 \begin{theorem}[\cite{hardin2005minimal}, \cite{hardin2014periodic}]\label{thm4} Let $\Lambda$ be a lattice in $\R^d$ with co-volume $|\Lambda|>0$. For $s>d$, there is a positive and finite constant $C_{s,d}$ such that
  \begin{align}\label{Csd1}
    \lim_{N\rightarrow \infty} \frac{\mathcal{E}_{s,\Lambda} (N)}{N^{1+s/d}}&=\lim_{N\rightarrow \infty} \frac{\mathcal{E}_{s}(\Omega_\Lambda; N)}{N^{1+s/d}}= {C_{s,d}}{|\Lambda|^{-s/d}}, \qquad s>d,\\
    \lim_{N\rightarrow \infty} \frac{\mathcal{E}_{d,\Lambda} (N)}{N^2\log N}&=\lim_{N\rightarrow \infty} \frac{\mathcal{E}_d(\Omega_\Lambda,N)}{N^2 \log N}=\frac{2\pi^{d/2}}{d\Gamma(\frac{d}{2})}|\Lambda|^{-1},\label{Cdd1}.
  \end{align}
\end{theorem}

By considering scaled lattice configurations (see Lemma~\ref{zetafunc}) of the form $\omega_{m^d}^{\Lambda}:=(1/m)\Lambda \cap \Omega_{\Lambda}$ for a lattice $\Lambda$
of co-volume 1, we obtain the following upper bound for $C_{s,d}$ that holds both for  $0<s<d$ and $s=\log$  where $C_{s,d}$ is as in Theorem~\ref{thm2}  as well as for  $s> d$ where $C_{s,d}$ is as in Theorem~\ref{thm2}.
\begin{corollary}
Let $\Lambda$ be a $d$-dimensional lattice with co-volume 1. Then,
\begin{equation}
C_{s,d}\le \begin{cases} \zeta_{\Lambda}(s), &s>0, s\neq d,\\ 2\zeta_{\Lambda}'(0), & s=\log.\end{cases}
\end{equation}
\end{corollary}

The constant $C_{s,d}$ for $s>d$ appearing in \eqref{Csd1} is known only in the case $d=1$  where  $C_{s,1}=\zeta_{\Z}(s)=2\zeta(s)$ and $\zeta(s)$ denotes the classical Riemann zeta function.
  For dimensions $d=2, 4, 8$, and $24$,  it has been conjectured (cf. \cite{cohn2007universally, brauchart2012next} and references therein) that $C_{s,d}$ for $s>d$ is also given by an Epstein zeta function, specifically, that  $C_{s,d}=\zeta_{\Lambda_d}(s)$ for $\Lambda_d$ denoting the  equilateral triangular (or hexagonal) lattice, the $D_4$ lattice, the $E_8$ lattice, and the Leech lattice (all scaled to have co-volume 1) in the dimensions $d=2, 4, 8,$ and 24, respectively.  In \cite{coulangeon2011energy}, it is shown that periodized lattice configurations   for these special lattices are {\em local} minima  of the energy for a large class of energy potentials that includes periodic Riesz $s$-energy potentials for $s>d$.  

Finally, we would like to say a little more about periodizing  Riesz potentials.     It is elementary to verify that
$$|x+v|^{-s}-|v|^{-s}=O\left(|v|^{-(s+1)}\right), \qquad (|v|\to \infty). $$
Let $C$ be a bounded set containing an open neighborhood of the origin and for $L>0$ let $C_L=LC$ (the set $C$ scaled by $L$).
Then,  we have
\begin{equation}\label{zetadiff}
\zeta(s;x)-\zeta(s)=|x|^{-s}+(1/2)\lim_{L\to \infty}\sum_{v\in  C_L\setminus\{0\}}|x+v|^{-s}-|v|^{-s}, \qquad (x\not\in \Lambda),
\end{equation}
for $s>d-1$.
Further assuming that $C$ is centrally symmetric (i.e., $v\in C\implies -v\in C$) and using the fact (see  \eqref{x+vx-v} Section~\ref{sect5}) that
\begin{equation*}
|x+v|^{-s}+|x-v|^{-s}-2|v|^{-s}=O\left(|v|^{-(s+2)}\right), \qquad (|v|\to \infty),
\end{equation*}
shows that \eqref{zetadiff} holds
for $s>d-2$ and so we obtain (up to a constant depending only on $s$) the same periodic Riesz potential   given in \eqref{FsLam}; i.e., the convergence factors procedure and the limit in \eqref{zetadiff} give the same result.  However, \eqref{zetadiff} breaks down for $0<s<d-2$.  In this case the energy is dominated by long range contributions  from translates near the boundary of $C$ and  \eqref{zetadiff} no longer holds.  In fact, as a consequence of Theorem~\ref{shellthm} below, the right hand side of \eqref{zetadiff} is $-\infty$ for all $x$.  For the case that $C=\mathcal{B}_d$, the unit ball in $\R^d$ centered at the origin, we find that \eqref{zetadiff} can be `renormalized' by dividing by
\begin{equation}\label{DLdef}
D_L:=\sum_{v\in  C_L\setminus\{0\}}|v|^{-s},
\end{equation} but  this leads to a non-periodic potential.

\begin{theorem}\label{shellthm}
Let $C=\mathcal{B}_d$, $0<s\le d-2$, and for $L>0$, let $D_L$ be given by \eqref{DLdef}.
Then
\begin{equation} \label{limslessd-2}
\lim_{L\to \infty}D_L^{-1}\left[ |x|^{-s}+(1/2)\sum_{v\in  C_L\setminus\{0\}}|x+v|^{-s}-|v|^{-s}\right]= -(s/d)(d-s-2)|x|^2,
\end{equation}
uniformly in $x$ on compact subsets of $ \R^d\setminus\Lambda$.
\end{theorem}

\noindent
{\bf Remarks:}
\begin{itemize}
\item Theorem~\ref{shellthm} implies that for an $N$-point configuration $\{x_1,\ldots, x_n\}$ the energy sum
\begin{equation}\label{L2infinity}
\begin{split}
\sum_{i\neq j} |x_i-x_j|^{-s}+\frac{1}{2}\sum_{v\in  C_L\setminus\{0\}}&(|x_i-x_j+v|^{-s}-|v|^{-s})\\ & =  -(s/d)(d-s-2)D_L\sum_{i\neq j} |x_i-x_j|^2+o(D_L).
\end{split}
\end{equation}
as $L\to \infty$.
\item Since the first term on the right side of \eqref{L2infinity} vanishes
if  $s=d-2$,  the dominant term of the left hand side of \eqref{L2infinity} is not determined by \eqref{limslessd-2} when $s=d-2$; i.e., all we know is that this term is $o(D_L)$.
\end{itemize}




\section{The Epstein Hurwitz Zeta Function}\label{ehzf}

In this section, we will review some relevant terminology and notation involving  special functions that will be crucial for our analysis in Section \ref{prf2}.

An argument  utilizing the integral representation of the Riesz kernel given in \eqref{eq2} together with the Poisson Summation Formula    can be used to establish the following lemma (cf., \cite[Section 1.4]{terras1985harmonic}).
\begin{lemma}\label{lemEpZeta1}
  The Epstein zeta function $\zeta_{\Lambda}(s)$ can be analytically continued to $\mathbb{C}\setminus\{d\}$ through the following formula :
  \begin{align*}
    \zeta_\Lambda(s)=\frac{2}{\Gamma(\frac{s}{2})}\left(\frac{2\pi^{\frac{d}{2}}|\Lambda|^{-1}}{s-d}-\frac{1}{s}\right)+\sum_{v\in\Lambda\setminus\{0\}}\int_1^\infty e^{-|v|^2t}\frac{t^{\frac{s}{2}-1}}{\Gamma(\frac{s}{2})}\mathrm{d}t+\frac{1}{|\Lambda|}\sum_{w\in\Lambda^*\setminus\{0\}}\int_0^1 \frac{\pi^{\frac{d}{2}}}{t^\frac{d}{2}}e^{-\frac{\pi^2|w|^2}{t}}\frac{t^{\frac{s}{2}-1}}{\Gamma(\frac{s}{2})}\mathrm{d}t.
  \end{align*}
\end{lemma}

\noindent
{\bf Remark.}
  From $\lim\limits_{s\to 0+}\Gamma(\frac{s}{2})=\infty$ and $\lim\limits_{s\to 0+}\frac{s}{2}\Gamma(\frac{s}{2})=\lim\limits_{s\to 0+}\Gamma(\frac{s}{2}+1)=\Gamma(1)=1$, it follows that $\zeta_\Lambda(0)=-1$ and $\zeta_{\Lambda}(0;x)\equiv0$ for any lattice $\Lambda$.\\

Next we establish the following relation between the Epstein and Epstein Hurwitz zeta functions.
\begin{lemma}\label{zetaidents}
  Let $\Lambda$ be a sublattice of $\Lambda'$. Then for any $s\in\C\setminus\{d\}$, it holds that
  \begin{equation}\label{EpHur1}
    \sum_{x \in \Lambda'\cap\Omega_\Lambda\setminus\{0\}}\zeta_\Lambda(s;x)=\zeta_{\Lambda'}(s)-\zeta_{\Lambda}(s).
  \end{equation}
\end{lemma}
\begin{proof}
  It is sufficient to prove that \eqref{EpHur1} holds for $s>d$, since the general result follows from the fact that both sides of this relation are analytic on $\C\setminus\{d\}$.   For $s>d$, we have by definition
  \begin{align*}
    \zeta_{\Lambda'}(s)&=\sum_{x\in\Lambda'\setminus\{0\}}\frac{1}{|x|^s}=\sum_{x \in \Lambda'\cap\Omega_\Lambda\setminus\{0\}}\sum_{v\in\Lambda}\frac{1}{|x+v|^s}+\sum_{v\in\Lambda\setminus\{0\}}\frac{1}{|v|^s}\\
    &=\sum_{x \in \Lambda'\cap\Omega_\Lambda\setminus\{0\}}\zeta_\Lambda(s;x)+\zeta_\Lambda(s),
  \end{align*}
  thus proving the lemma.
\end{proof}
Using the above lemma and scaling properties of Epstein zeta functions we obtain the following:
\begin{lemma}\label{zetafunc}
For every $m\in\N$ and $s\in\C\setminus\{d\}$, it holds that
\begin{align}\label{zetam}
  \sum_{x \in \frac{1}{m}\Lambda\cap\Omega_\Lambda\setminus\{0\}}\zeta_\Lambda(s;x)=(m^s-1)\zeta_\Lambda(s).
\end{align}
Therefore,
\begin{align}
\sum_{\substack{x, y\in \frac{1}{m}\Lambda\cap\Omega_\Lambda\\x\neq y}}\zeta_\Lambda(s;x-y)=m^d(m^s-1)\zeta_\Lambda(s).
\end{align}
\end{lemma}

We will also require the following two lemmas, which establish continuity properties of the Epstein Hurwitz Zeta function with respect to the lattice.

\begin{lemma}\label{lem16}
   Let $\{P_m\}_{m\in\mathbb{N}}$ be a sequence of $d\times d$ matrices such that $P_m\rightarrow \mbox{P}$ in norm as $m\rightarrow\infty$. Fix any distinct $x$ and $y$ in $\Omega_\Lambda$ and suppose $\{x_m\}_{m\in\mathbb{N}}$ and $\{y_m\}_{m\in\mathbb{N}}$ are sequences in $\Omega_\Lambda$ converging to $x$ and $y$, respectively. Then for any compact set $K\subset\mathbb{C}\setminus\{d\}$, $\zeta_{P_m\Lambda}(s;P_m(x_m-y_m))$ converges to $\zeta_{P\Lambda}(s;P(x-y))$
   and $\zeta_{P_m\Lambda}'(s;P_m(x_m-y_m))$ converges to $\zeta_{P\Lambda}'(s;P(x-y))$ uniformly for $s$ in $K$ as $m\to\infty$.
\end{lemma}

\begin{proof}
Let $R=\sup_{s\in K} \Ree (s)$ and $r=\inf_{s\in K}\Ree(s)$.
Notice that $\sup_{s\in K}|1/\Gamma(\frac{s}{2})|$ is finite since $1/\Gamma(\frac{s}{2})$ is entire.
 Let $m$ be large enough so that $x_m-y_m\not\in\Lambda$.  Using \eqref{Fsdef}, we have
\begin{equation}\label{lem7proof1}
\begin{split}
 |\zeta_{P_m\Lambda}&(s;P_m(x_m-y_m))-\zeta_{P\Lambda}(s;P(x-y))|\\
 &=|F_{s,P_m\Lambda} (P_m x_m-P_m y_m)-F_{s,P\Lambda} (Px-Py)|\\
&\leq \int_1^\infty \sum_{v\in\Lambda}\left |e^{-|P_m(x_m-y_m+v)|^2t}-e^{-|P(x-y+v)|^2t}\right|\frac{|t^{\frac{s}{2}-1}|}{|\Gamma(\frac{s}{2})|}\mathrm{d}t\\
&\qquad\qquad+\int_0^1 \sum_{w\in\Lambda^*\setminus\{0\}}\left| e^{2\pi i w\cdot (x_m-y_m)}  e^{-\frac{\pi^2|(P_m^{-1})^T w|^2}{t}}-e^{2\pi i w\cdot (x-y)}e^{-\frac{\pi^2| (P^{-1})^Tw|^2}{t}}\right|\frac{\pi^\frac{d}{2}|t^{\frac{s-d}{2}-1}|}{|\Gamma(\frac{s}{2})|}\mathrm{d}t\\
&\leq \int_1^\infty \sum_{v\in\Lambda}\left |e^{-|P_m(x_m-y_m+v)|^2t}-e^{-|P(x-y+v)|^2t}\right|\frac{t^{\frac{R}{2}-1}}{\inf_{s\in K}|\Gamma(\frac{s}{2})|}\mathrm{d}t\\
&\qquad\qquad+\int_0^1 \sum_{w\in\Lambda^*\setminus\{0\}}\left| e^{2\pi i w\cdot (x_m-y_m)}  e^{-\frac{\pi^2|(P_m^{-1})^T w|^2}{t}}-e^{2\pi i w\cdot (x-y)}e^{-\frac{\pi^2| (P^{-1})^Tw|^2}{t}}\right|\frac{\pi^\frac{d}{2}t^{\frac{r-d}{2}-1}}{\inf_{s\in K}|\Gamma(\frac{s}{2})|}\mathrm{d}t.
\end{split}
\end{equation}
As in \cite{hardin2014periodic}, it is elementary to establish that integrals of the form
$$
\int_1^\infty \sum_{v\in\Lambda} e^{-|P(x-y+v)|^2t}  {t^{\frac{R}{2}-1}}\mathrm{d}t \text{ and }
\int_0^1 \sum_{w\in\Lambda^*\setminus\{0\}}  e^{-\frac{\pi^2| (P^{-1})^Tw|^2}{t}}{\pi^\frac{d}{2}t^{\frac{r-d}{2}-1}}\mathrm{d}t$$
are finite and thus, by dominated convergence, it follows that the expressions in \eqref{lem7proof1}
  tend to zero as $m\rightarrow\infty$.

 Cauchy's integral formula for derivatives then implies that
$\zeta'_{P_m\Lambda}(s;P_m(x_m-y_m))\rightarrow \zeta'_{P\Lambda}(s;P(x-y))$ uniformly for $s\in K$ as $m\to \infty$.
\end{proof}

We remark that the proof of Lemma~\ref{lem16} shows that $F_{s,P_m\Lambda} (P_m x_m-P_m y_m)$ converges  to $F_{s,P\Lambda} (Px-Py)$ as $m\to \infty$
uniformly for $s$ in any compact set of $\C$.

\begin{corollary}\label{CorMinEn}
Let $\{P_m\}_{m\in\mathbb{N}}$ be a sequence of $d\times d$ matrices such that $P_m\rightarrow \mbox{P}$ in norm as $m\rightarrow\infty$ and suppose $s>0$ or $s=\log$.   Then,  for all $N\ge 2$, we have $\mathcal{E}_{s,P_m\Lambda}(N)\to \mathcal{E}_{s,P\Lambda}(N)$ as $m\to \infty$.
\end{corollary}
\begin{proof}
Let $\omega_N^*\subset \Omega_\Lambda$ be such that $P\omega_N^*$ is an  ${E}_{s,P\Lambda}$ optimal  $N$-point configuration.  Then,
$$
\limsup_{m\to \infty} \mathcal{E}_{s,P_m\Lambda}(N)\le \limsup_{m\to \infty}  {E}_{s,P_m\Lambda}(P_m\omega_N^*)={E}_{s,P\Lambda}(P\omega_N^*)
= \mathcal{E}_{s,P\Lambda}(N),
$$
where the next to last equality follows from Lemma~\ref{lem16}.

Next let $\omega_N^m=\{x_1^m,\ldots,x_N^m\}\subset \Omega_\Lambda$ be such that $P_m\omega_N^m$ is an optimal $N$-point configuration for ${F}_{s,P_m\Lambda}$.  Let  $\{\omega_N^{m_k}\}_{k\in \N}$ be a subsequence such that $$\lim_{k\to \infty}{E}_{s,P_{m_k}\Lambda}(P_{m_k}\omega_N^{m_k})= \liminf_{m\to \infty} \mathcal{E}_{s,P_m\Lambda}(N).$$  Using the  compactness of $\Omega_\Lambda$ in the `flat torus' topology, we may assume without loss of generality that  $\{\omega_N^{m_k}\}_{k\in \N}$ converges to some $N$-point configuration $\tilde{\omega}_N=\{\tilde x_1,\ldots,\tilde x_N\}$; i.e., $x_j^{m_k}\to \tilde x_j$ as $k\to \infty$ for each $j=1,\ldots, N$.
Then we have
$$
\liminf_{m\to \infty} \mathcal{E}_{s,P_m\Lambda}(N)= \lim_{k\to \infty}  {E}_{s,P_{m_k}\Lambda}(P_{m_k}\omega_N^{m_k}) = {E}_{s,P\Lambda}(P\tilde{\omega}_N)
\ge \mathcal{E}_{s,P\Lambda}(N),
$$
where the next to last equality follows from Lemma~\ref{lem16}.
\end{proof}
Finally,  the following result expresses  continuity properties of the Epstein zeta function with respect to the lattice similar to the results in Lemma~\ref{lem16} for the Epstein Hurwitz zeta function.

\begin{lemma}\label{lem17}
   Let $\{P_m\}_{m\in\mathbb{N}}$ be a sequence of $d\times d$ matrices such that $P_m\rightarrow \mbox{P}$ in norm as $m\rightarrow\infty$.  Then for any compact set $K\subset\mathbb{C}\setminus\{d\}$,
 $\zeta_{P_m\Lambda}(s)$ converges to $\zeta_{P\Lambda}(s)$ uniformly in $K$ and hence $\zeta'_{P_m\Lambda}(s)\rightarrow \zeta'_{P\Lambda}(s)$ for all $s\in\mathbb{C}\setminus\{d\}$
as $m\to \infty$.
 \end{lemma}
 \begin{proof}
 Using Lemma~\ref{lemEpZeta1},  a similar argument as in the proof of Lemma \ref{lem16} implies that $\zeta_{P_m\Lambda}(s)$ converges uniformly to $\zeta_{P\Lambda}(s)$ on compact sets $K\subset \mathbb{C}\setminus\{d\}$.  The convergence of the derivatives then follows from Cauchy's integral formula for derivatives.
 \end{proof}


\section{Proof of Theorem~\ref{thm2} \label{prf2}}

Throughout this section and the next we shall assume that $\Lambda=A\Z^d$ denotes a $d$-dimensional lattice in $\R^d$ with fundamental domain $\Omega=\Omega_{\Lambda}$, co-volume 1, and generating matrix $A$.  Then Theorem~\ref{thm2} follows from a simple rescaling.
We shall find it convenient  to use what we call the
{\em classical periodic Riesz $s$-potential} $F_{s,\Lambda}^{\rm cp}(x):=\zeta_{\Lambda}(s;x)$ which, for $s\neq d$, differs from $F_{s,\Lambda} $ only by the constant $\frac{2\pi^{\frac{d}{2}}}{\Gamma(\frac{s}{2})(d-s)}$.  Similarly, we call  $F_{\log,\Lambda}^{\rm cp}(x):=2\zeta_{\Lambda}'(0;x)$ the
{\em classical periodic logarithmic potential}.  The  energies associated with these potentials are given by
 \begin{equation}\label{Escpdef}
 E_{s,\Lambda}^{\rm cp} (\omega_N):=\sum_{j\neq k}\zeta_{\Lambda}(s;x_j-x_k), \qquad (s>0),
\end{equation}
and, similarly,
 \begin{equation}\label{Elogcpdef}
 E_{\log,\Lambda}^{\rm cp} (\omega_N):=2\sum_{j\neq k}\zeta_{\Lambda}'(0;x_j-x_k),
\end{equation}
and we denote the respective minimal $N$-point energies by $\mathcal{E}_{s,\Lambda}^{\rm cp}(N)$ and $\mathcal{E}_{\log,\Lambda}^{\rm cp}(N)$.

From \eqref{FsLam}, we obtain
\begin{equation}
  \mathcal{E}_{s,\Lambda} (N)= \mathcal{E}_{s,\Lambda}^{\rm cp}(N)+\frac{2\pi^{\frac{d}{2}}}{\Gamma(\frac{s}{2})(d-s)}N(N-1),  \label{eq27}
\end{equation}
and \begin{equation}
  \mathcal{E}_{\log,\Lambda} (N)= \mathcal{E}_{\log,\Lambda}^{\rm cp}(N)+\frac{2\pi^{\frac{d}{2}}}{d}N(N-1),  \label{eq28}
\end{equation}

Define
\begin{align*}
  \underline{g}_{s,d}(\Lambda)&:=\liminf_{N\rightarrow \infty} \frac{\mathcal{E}_{s,\Lambda}^{\rm cp}(N)}{N^{1+s/d}},\\ \overline{g}_{s,d}(\Lambda)&:=\limsup_{N\rightarrow \infty} \frac{\mathcal{E}_{s,\Lambda}^{\rm cp}(N)}{N^{1+s/d}},\\
  \underline{g}_{\log,d}(\Lambda)&:=\liminf_{N\rightarrow \infty} \frac{\mathcal{E}_{\log,\Lambda}^{\rm cp}(N)+\frac{2}{d}N\log N}{N},\\ \overline{g}_{\log,d}(\Lambda)&:=\limsup_{N\rightarrow \infty} \frac{\mathcal{E}_{\log,\Lambda}^{\rm cp}(N)+\frac{2}{d}N\log N}{N}.
\end{align*}
Our use of these quantities is motivated by the proof of the main results in \cite{hardin2005minimal}, and indeed the general strategy of our proofs is similar to that of \cite{hardin2005minimal}.  More precisely, we shall prove $\underline{g}_{s,d}(\Lambda)=\overline{g}_{s,d}(\Lambda)$ and $\underline{g}_{\log,d}(\Lambda)=\overline{g}_{\log,d}(\Lambda)$ and that these limits are finite.
We first need estimates on quantities appearing in  \eqref{Fsdef} and \eqref{Flogdef}.

\begin{lemma}\label{lem15} Let $s>0$ and $\Lambda$ be a $d$-dimensional lattice with co-volume 1 and $l_0:=\min\limits_{0\neq v\in \Lambda}\{|v|\}$. The following relations hold.
  \begin{equation}\label{lem15a}
\sum_{w\in\Lambda^*} e^{-\frac{\pi^2|w|^2}{t}}t^{-\frac{d}{2}}=\pi^{-\frac{d}{2}}+O(e^{-l_0^2 t}),\qquad \text{as }t\rightarrow \infty,
\end{equation}
\begin{equation}\label{lem15b}
  \sum_{w\in\Lambda^*}\int_1^{\frac{1}{\delta}}e^{-\frac{\pi^2|w|^2}{t}}t^{\frac{s-d}{2}-1}\mathrm{d}t=\frac{2\pi^{-\frac{d}{2}}}{s}\delta^{-\frac{s}{2}}+O(1), \qquad  \text{as $\delta\rightarrow 0^+$},
  \end{equation}
  \begin{equation}\label{lem15c}
\sum_{w\in\Lambda^*}\int_1^{\frac{1}{\delta}}e^{-\frac{\pi^2|w|^2}{t}}t^{-\frac{d}{2}-1}\mathrm{d}t=\pi^{-\frac{d}{2}}\log\delta^{-1}+O(1), \qquad  \text{as $\delta\rightarrow 0^+$}\end{equation}
\end{lemma}
\begin{proof}
  Applying Poisson Summation, we obtain
\[
\sum_{w\in\Lambda^*} e^{-\frac{\pi^2|w|^2}{t}}t^{-\frac{d}{2}}=\pi^{-d/2}\sum_{v\in\Lambda} e^{-|v|^2 t}=\pi^{-d/2}+\pi^{-d/2}e^{-l_0^2 t}\sum_{v\in\Lambda\setminus\{0\}} e^{-(|v|^2-l_0^2)t}=\pi^{-d/2}+O(e^{-l_0^2 t}),
\]
proving  \eqref{lem15a}.
 Hence, there exists a constant $C_1$ such that
$$
\left|\sum_{w\in\Lambda^*} e^{-\frac{\pi^2|w|^2}{t}}t^{-\frac{d}{2}}-\pi^{-\frac{d}{2}}\right|\leq C_1e^{-l_0^2 t},
$$
and so, multiplying both sides of the above by $t^{\frac{s}{2}-1}$, we have
$$
\left|\sum_{w\in\Lambda^*} e^{-\frac{\pi^2|w|^2}{t}}t^{\frac{s-d}{2}-1}-\pi^{-\frac{d}{2}}t^{\frac{s}{2}-1}\right|\leq C_1t^{\frac{s}{2}-1}e^{-l_0^2 t}
$$
and so
\begin{equation}\label{lem15proof1}
\begin{split} &\left|\int_1^{\frac{1}{\delta}}\left( \sum_{w\in\Lambda^*} e^{-\frac{\pi^2|w|^2}{t}}t^{\frac{s-d}{2}-1}-\pi^{-\frac{d}{2}}t^{\frac{s}{2}-1}\right)\mathrm{d}t\right|\leq \int_1^{\frac{1}{\delta}}C_1t^{\frac{s}{2}-1}e^{-l_0^2 t}\mathrm{d}t\\ &\qquad \leq \int_1^{\infty}C_1t^{\frac{s}{2}-1}e^{-l_0^2 t}\mathrm{d}t=:C_2(s).
\end{split}
\end{equation}
Therefore,
$$\left|\sum_{w\in\Lambda^*}\int_1^{\frac{1}{\delta}}e^{-\frac{\pi^2|w|^2}{t}}t^{\frac{s-d}{2}-1}\mathrm{d}t-\frac{2\pi^{-\frac{d}{2}}}{s}(\delta^{-\frac{s}{2}}-1)\right|\leq C_2(s), \qquad s>0
$$
proving \eqref{lem15b}, while substituting $s=0$ into \eqref{lem15proof1} yields
$$\left|\sum_{w\in\Lambda^*}\int_1^{\frac{1}{\delta}}e^{-\frac{\pi^2|w|^2}{t}}t^{-\frac{d}{2}-1}\mathrm{d}t-\pi^{-\frac{d}{2}}\log\delta^{-1}\right|\leq C_2(0), \qquad s=0.
$$
 proving  \eqref{lem15c}.
\end{proof}

The following lemma is the key calculation that allows us to apply the method of \cite{hardin2005minimal}.  Once we have established this lemma, the only remaining technical difficulty will be to establish the fact that the constants $C_{s,d}$ and $C_{\log,d}$ are independent of the lattice $\Lambda$.

\begin{lemma}\label{finlims}
With $\Lambda$ as in Lemma~\ref{lem15} and $s>0$, the following inequalities hold:
\begin{align*}
-\infty<\underline{g}_{s,d}(\Lambda)&\leq\overline{g}_{s,d}(\Lambda)\leq \zeta_\Lambda(s)<\infty, \\
-\infty<\underline{g}_{\log,d}(\Lambda)&\leq\overline{g}_{\log,d}(\Lambda)\leq 0.
\end{align*}
\end{lemma}

\begin{proof}
Let us first consider the case  $s>0$.
For any configuration $\omega_N=(x_j)_{j=1}^N$ in $\Omega_\Lambda$ and any $\delta\in(0,1]$,
\begin{align*}
  E_{s,\Lambda} (\omega_N)&=\sum_{j\neq k}K_{s,\Lambda} (x_j,x_k)=:I_1+I_2.
\end{align*}
where
\begin{align*}
  I_1&=\sum_{j\neq k}\sum_{v\in\Lambda}\int_1^\infty e^{-|x_j-x_k+v|^2t}\frac{t^{\frac{s}{2}-1}}{\Gamma(\frac{s}{2})}\mathrm{d}t,\\
  I_2&=\sum_{j\neq k}\sum_{w\in\Lambda^*\setminus\{0\}}e^{2\pi i w\cdot (x_j-x_k)}\int_0^1 \frac{\pi^{\frac{d}{2}}}{t^\frac{d}{2}}e^{-\frac{\pi^2|w|^2}{t}}\frac{t^{\frac{s}{2}-1}}{\Gamma(\frac{s}{2})}\mathrm{d}t.
\end{align*}
Let
\begin{align*}
  h_{\delta}(x):=\int_1^{\frac{1}{\delta}} e^{-|x|^2t}\frac{t^{\frac{s}{2}-1}}{\Gamma(\frac{s}{2})}\mathrm{d}t.
\end{align*}
then
\begin{align*}
  \hat{h}_{\delta}(\xi)=\int_1^{\frac{1}{\delta}} \left(\frac{\pi}{t}\right)^{\frac{d}{2}} e^{-\frac{\pi^2|\xi|^2}{t}} \frac{t^{\frac{s}{2}-1}}{\Gamma(\frac{s}{2})}\mathrm{d}t\geq 0,\quad \hat{h}_{\delta}(0)=\frac{2\pi^{\frac{d}{2}}}{\Gamma(\frac{s}{2})(d-s)}\left(1-\delta^{\frac{d-s}{2}}\right).
\end{align*}
Notice that since the upper limit of the integral defining $h_{\delta}$ is finite, it is easy to verify that $h_{\delta}$ satisfies the hypotheses of Poisson Summation.  Applying it gives us the following inequalities:
\begin{equation}\label{eq19}\begin{split}
I_1&\geq \sum_{j\neq k}\sum_{v\in\Lambda} h_{\delta}(x_j-x_k+v) \\
&=\sum_{j\neq k}\sum_{w\in\Lambda^*} \hat{h}_{\delta}(w)e^{2\pi i w\cdot(x_j-x_k)}\\
&=\sum_{w\in\Lambda^*} \hat{h}_{\delta}(w)\left(\sum_{j,k}e^{2\pi i w\cdot(x_j-x_k)}-N\right)\\
&=\sum_{w\in\Lambda^*} \hat{h}_{\delta}(w)\left(\left|\sum_{j}e^{2\pi i w\cdot x_j}\right|^2-N\right)\\
&\geq N^2\hat{h}_{\delta}(0)-N\sum_{w\in\Lambda^*}\hat{h}_{\delta}(w)\\
&=N^2\frac{2\pi^{\frac{d}{2}}}{\Gamma(\frac{s}{2})(d-s)}\left(1-\delta^{\frac{d-s}{2}}\right)-N\frac{\pi^{\frac{d}{2}}}{\Gamma(\frac{s}{2})}\sum_{w\in\Lambda^*}\int_1^{\frac{1}{\delta}} e^{-\frac{\pi^2|w|^2}{t}}t^{\frac{s-d}{2}-1}\mathrm{d}t
\end{split}\end{equation}
By Lemma \ref{lem15}, we conclude
\begin{equation}\label{eq19b}
\begin{split}
  I_1&\geq N^2\frac{2\pi^{\frac{d}{2}}}{\Gamma(\frac{s}{2})(d-s)}\left(1-\delta^{\frac{d-s}{2}}\right)-N\frac{2\pi^{\frac{d}{2}}}{s\Gamma(\frac{s}{2})}\left(\pi^{-\frac{d}{2}}\delta^{-\frac{s}{2}}+O(1)\right) \\
&=\frac{2\pi^{\frac{d}{2}}}{\Gamma(\frac{s}{2})(d-s)}N^2-\frac{2\pi^{\frac{d}{2}}}{\Gamma(\frac{s}{2})(d-s)}N^2\delta^{\frac{d-s}{2}}-\frac{2}{s\Gamma(\frac{s}{2})}N\delta^{-\frac{s}{2}}-O(N),
\end{split}
\end{equation}

To obtain lower bounds on $I_2$, we calculate
\begin{equation}\label{eq20} \begin{split}
     I_2&=\sum_{w\in\Lambda^*\setminus\{0\}}\left(\sum_{j,k}e^{2\pi i w\cdot (x_j-x_k)}-N\right)\int_0^1 \frac{\pi^{\frac{d}{2}}}{t^\frac{d}{2}}e^{-\frac{\pi^2|w|^2}{t}}\frac{t^{\frac{s}{2}-1}}{\Gamma(\frac{s}{2})}\mathrm{d}t  \\
  &=\sum_{w\in\Lambda^*\setminus\{0\}}\left(\left|\sum_{j}e^{2\pi i w\cdot x_j}\right|^2-N\right)\int_0^1 \frac{\pi^{\frac{d}{2}}}{t^\frac{d}{2}}e^{-\frac{\pi^2|w|^2}{t}}\frac{t^{\frac{s}{2}-1}}{\Gamma(\frac{s}{2})}\mathrm{d}t \\
  &\geq -N\cdot  \frac{\pi^{\frac{d}{2}}}{\Gamma(\frac{s}{2})}\sum_{w\in\Lambda^*\setminus\{0\}}\int_0^1e^{-\frac{\pi^2|w|^2}{t}}t^{\frac{s-d}{2}-1}\mathrm{d}t \\
  &=O(N).
 \end{split}
 \end{equation}
 Therefore
 \begin{equation*}
   \begin{split}
     E_{s,\Lambda} (\omega_N)&=I_1+I_2\geq\frac{2\pi^{\frac{d}{2}}}{\Gamma(\frac{s}{2})(d-s)}N^2-\frac{2\pi^{\frac{d}{2}}}{\Gamma(\frac{s}{2})(d-s)}N^2\delta^{\frac{d-s}{2}}-\frac{2}{s\Gamma(\frac{s}{2})}N\delta^{-\frac{s}{2}}-O(N).
   \end{split}
 \end{equation*}
 If we let $\delta=\pi^{-1}N^{-\frac{2}{d}}$, then this lower bound becomes
 \begin{equation} \label{eq4}
   E_{s,\Lambda} (\omega_N)\geq\frac{2\pi^{\frac{d}{2}}}{\Gamma(\frac{s}{2})(d-s)}N^2+C^*N^{1+\frac{s}{d}}+O(N).
 \end{equation}
where
\begin{align*}
  C^*=-\frac{2\pi^{\frac{s}{2}}d}{\Gamma(\frac{s}{2})s(d-s)}.
\end{align*}
 The right hand side of (\ref{eq4}) is independent of $\omega_N$ and thus
 \begin{align*}
   \mathcal{E}_{s,\Lambda} (N)&\geq\frac{2\pi^{\frac{d}{2}}}{\Gamma(\frac{s}{2})(d-s)}N^2+C^*N^{1+\frac{s}{d}}+O(N),\\
   \mathcal{E}_{s,\Lambda}^{\rm cp}(N)&\geq C^*N^{1+\frac{s}{d}}+O(N).
 \end{align*}
We conclude that
 $\underline{g}_{s,d}(\Lambda)\geq C^*$.

 To establish the finiteness of $\overline{g}_{s,d}$, we will use the same method as was used in \cite{hardin2005minimal}.  For any natural number $N$, let $m=m_N$ be a positive integer such that $(m-1)^d<N\leq m^d$. Let $\omega^m=\frac{1}{m}\Lambda\cap\Omega_\Lambda$. Then
 \begin{equation}\label{eq23}
   \mathcal{E}_{s,\Lambda}^{\rm cp}(m^d)\leq E_{s,\Lambda}^{\rm cp}(\omega^m)=\sum_{\substack{x_j,x_k\in \frac{1}{m}\Lambda\cap\Omega_\Lambda\\x_j\neq x_k}}\zeta_\Lambda(s;x_j-x_k)=m^d(m^s-1)\zeta_\Lambda(s),
 \end{equation}
where we used Lemma \ref{zetafunc}

 \vspace{2mm}

 As $\left\{\frac{\mathcal{E}_{s,\Lambda}^{\rm cp}(N)}{N(N-1)}\right\}_{N=2}^\infty$ is an increasing sequence (see, e.g., \cite[Chapter II \S 3.12, page 160]{landkof1972foundations}) we arrive at the following:
 \begin{align*}
   \overline{g}_{s,d}(\Lambda)&=\limsup_{N\rightarrow \infty} \frac{\mathcal{E}_{s,\Lambda}^{\rm cp}(N)}{N^{1+s/d}}
=\limsup_{N\rightarrow \infty} \frac{\mathcal{E}_{s,\Lambda}^{\rm cp}(N)}{N(N-1)}\cdot \frac{N-1}{N^{\frac{s}{d}}}\\
&\leq\limsup_{N\rightarrow \infty} \frac{\mathcal{E}_{s,\Lambda}^{\rm cp}(\Omega_\Lambda,m^d)}{m^d(m^d-1)}\cdot \frac{N-1}{N^{\frac{s}{d}}}\\
&\leq\limsup_{N\rightarrow \infty} \frac{m^d(m^s-1)\zeta_\Lambda(s)}{m^d(m^d-1)}\cdot \frac{N-1}{N^{\frac{s}{d}}}=\zeta_\Lambda(s)<\infty.
 \end{align*}

Now we turn our attention to the classical periodic logarithmic energy.  Using (\ref{Flog2}), (\ref{eq19b}), and (\ref{eq20}), we obtain
\begin{align*}
  E_{\log,\Lambda} (\omega_N)&=\lim_{s\rightarrow 0^+}\Gamma\left(\frac{s}{2}\right)E_{s,\Lambda} (\omega_N)=\lim_{s\rightarrow 0^+}\Gamma\left(\frac{s}{2}\right)(I_1+I_2)\\
&\geq N^2\frac{2\pi^{\frac{d}{2}}}{d}\left(1-\delta^{\frac{d}{2}}\right)-N\pi^{\frac{d}{2}}\sum_{w\in\Lambda^*}\int_1^{\frac{1}{\delta}} e^{-\frac{\pi^2|w|^2}{t}}t^{-\frac{d}{2}-1}\mathrm{d}t +O(N)\\
&=N^2\frac{2\pi^{\frac{d}{2}}}{d}\left(1-\delta^{\frac{d}{2}}\right)-N\pi^{\frac{d}{2}}\left(\pi^{-\frac{d}{2}}\log\delta^{-1}+O(1)\right) +O(N)\\
&=\frac{2\pi^{\frac{d}{2}}}{d}N^2-\frac{2\pi^{\frac{d}{2}}}{d}N^2\delta^{\frac{d}{2}}-N\log\delta^{-1}+O(N).
\end{align*}
If we let $\delta=N^{-\frac{2}{d}}$, then we get
\begin{align*}
E_{\log,\Lambda} (\omega_N)=\frac{2\pi^{\frac{d}{2}}}{d}N^2-\frac{2}{d}N\log N+O(N).
\end{align*}
Thus
\begin{align*}
  &\mathcal{E}_{\log,\Lambda} (N)\geq \frac{2\pi^{\frac{d}{2}}}{d}N^2-\frac{2}{d}N\log N+O(N),\\
&\mathcal{E}_{\log,\Lambda}^{\rm cp}(N)+\frac{2}{d}N\log N\geq O(N),
\end{align*}
and we conclude that $\underline{g}_{\log,d}(\Lambda)>-\infty$.

To establish the finiteness of $\overline{g}_{s,d}$, let $m=m_N$ be a positive integer such that $(m-1)^d<N\leq m^d$. Let $\omega^m=\frac{1}{m}\Lambda\cap\Omega_\Lambda$. Then by (\ref{eq23})
\begin{align*}
  E_{s,\Lambda}^{\rm cp}(\omega^m)=m^d(m^s-1)\zeta_\Lambda(s).
\end{align*}
By definition,
\begin{align*}
  \mathcal{E}_{\log,\Lambda}^{\rm cp}(m^d)\leq E_{\log,\Lambda}^{\rm cp}(\omega^m)&=2\left.\frac{\mathrm{d}}{\mathrm{d}s}E_{s,\Lambda}^{\rm cp}(\omega_N)\right|_{s=0}\\
&=2\left. m^d\left(m^s\log m\cdot \zeta_\Lambda(s)+(m^s-1)\zeta'(s)\right)\right|_{s=0}\\
&=2m^d\log m\cdot\zeta_\Lambda(0)=-2m^d\log m=-\frac{2}{d}m^d\log m^d
\end{align*}
Here we use the fact that $\zeta_\Lambda(0)=-1$ for every lattice $\Lambda$ (see the remark following Lemma~\ref{lemEpZeta1}).
We conclude that
 \begin{align*}
\frac{\mathcal{E}_{s,\Lambda}^{\rm cp}(N)}{N}&=\frac{\mathcal{E}_{s,\Lambda}^{\rm cp}(N)}{N(N-1)}\cdot(N-1)\leq\frac{\mathcal{E}_{s,\Lambda}^{\rm cp}(m^d)}{m^d(m^d-1)}\cdot(N-1)\leq -\frac{2}{d}\log m^d \cdot \frac{N-1}{m^d-1}
\end{align*}
This implies
\begin{align*}
\frac{\mathcal{E}_{s,\Lambda}^{\rm cp}(N)+\frac{2}{d}N\log N}{N}&\leq -\frac{2}{d}\log m^d \cdot \frac{N-1}{m^d-1}+\frac{2}{d}\log N,
\end{align*}
which tends to $0$ as $\bnri$, and hence $\overline{g}_{\log,d}(\Lambda)\leq 0$.
\end{proof}

The following lemma establishes scaling properties of the classical periodic energy and will be helpful in establishing independence of the constants $C_{s,d}$ and $C_{\log,d}$ of the lattice $\Lambda$.

\begin{lemma}\label{lem1}
Let $\Lambda$ be a lattice and $\Lambda'=B\Lambda$ be a sublattice of $\Lambda$ (i.e. $B\in GL(d,\mathbb{Z})$), then for any $N>0$,
  \begin{align*}
    \mathcal{E}_{s,\Lambda'}^{\rm cp}( N|\det B|)&\leq|\det B|  \mathcal{E}_{s,\Lambda}^{\rm cp}(N)+N|\det B|(\zeta_\Lambda(s)-\zeta_{\Lambda'}(s)),\\
\mathcal{E}_{\log,\Lambda'}^{\rm cp}(N|\det B|)&\leq|\det B|  \mathcal{E}_{\log,\Lambda}^{\rm cp}(N)+2N|\det B|(\zeta_\Lambda'(0)-\zeta_{\Lambda'}'(0)).
  \end{align*}
\end{lemma}


\begin{proof}
  For any $\omega_N=(x_j)_{j=1}^N \in(\Omega_{\Lambda})^N$, let $S(\omega_N)=(\omega_N+\Lambda)\cap\Omega_{\Lambda'}$. Then $S(\omega_N)$ is a $(N|\det B|)$-point configuration in $\Omega_{\Lambda'}$ and
  \begin{align}
    E_{s,\Lambda'}^{\rm cp}(S(\omega_N))&=\sum_{\substack{x,y\in S(\omega_N)\\x\neq y}} \zeta_{\Lambda'}(s,x-y)=\sum_{j,k}\sum_{\substack{r \in \Lambda, x_j+r\in\Omega_{\Lambda'}\\t \in \Lambda, x_k+t\in\Omega_{\Lambda'}\\x_j+r\neq x_k+t}}\zeta_{\Lambda'}(s;x_j+r-x_k-t)\nonumber\\
    &=\sum_{j\neq k}\sum_{\substack{r \in \Lambda\cap\Omega_{\Lambda'}\\t \in \Lambda\cap\Omega_{\Lambda'}}}\zeta_{\Lambda'}(s;x_j-x_k+r-t)+\sum_{j=1}^N\sum_{\substack{r \in \Lambda\cap\Omega_{\Lambda'}\\t \in \Lambda\cap\Omega_{\Lambda'}\\r\neq t}}\zeta_{\Lambda'}(s;r-t)\nonumber\\
    &=\sum_{j\neq k}\sum_{r \in \Lambda\cap\Omega_{\Lambda'}}\zeta_\Lambda(s;x_j-x_k)+N\sum_{r \in \Lambda\cap\Omega_{\Lambda'}}(\zeta_\Lambda(s)-\zeta_{\Lambda'}(s))\nonumber\\
    &=|\det B| \cdot E_{s,\Lambda}^{\rm cp}(\omega_N)+N|\det B|(\zeta_\Lambda(s)-\zeta_{\Lambda'}(s)), \label{eq24}
  \end{align}
where we used Lemma \ref{zetaidents}.  Taking the infimum over all configurations $(x_j)_{j=1}^N \in(\Omega_{\Lambda})^N$, we conclude that
  \begin{align*}
    \mathcal{E}_{s,\Lambda'}^{\rm cp}(N|\det B|)&\leq \inf_{\omega_N\in(\Omega_{\Lambda})^N} E_{s,\Lambda'}^{\rm cp}(S(\omega_N))=|\det B| \cdot \mathcal{E}_{s,\Lambda}^{\rm cp}(N)+N|\det B|(\zeta_\Lambda(s)-\zeta_{\Lambda'}(s)).
  \end{align*}
The logarithmic case follows from this by differentiating (\ref{eq24}) and evaluating at $s=0$.
\end{proof}

\begin{corollary}\label{cor1}
  For any positive integers $m$ and $N$, we have
  \begin{align} \label{cor1eq1}
    &\frac{\mathcal{E}_{s,\Lambda}^{\rm cp}(m^d N)}{(m^d N)^{1+\frac{s}{d}}}\leq \frac{\mathcal{E}_{s,\Lambda}^{\rm cp}( N)}{N^{1+\frac{s}{d}}}+\frac{(1-m^{-s})\zeta_\Lambda(s)}{N^\frac{s}{d}},\\\label{cor1eq2}
&\frac{\mathcal{E}_{\log,\Lambda}^{\rm cp}(m^dN)+\frac{2}{d}m^dN\log(m^dN)}{m^d N}\leq \frac{\mathcal{E}_{\log,\Lambda}^{\rm cp}(N)+\frac{2}{d}N\log N}{N}.
  \end{align}
\end{corollary}

\begin{proof}
Using Lemma \ref{lem1}, we obtain
  \begin{align*}
    \mathcal{E}_{s,m\Lambda}^{\rm cp}(m^d N)\leq m^d \cdot \mathcal{E}_{s,\Lambda}^{\rm cp}(N)+m^d N\left(\zeta_\Lambda(s)-\zeta_{m \Lambda}(s)\right),\\
\mathcal{E}_{\log,m\Lambda}^{\rm cp}(m^d N)\leq m^d \cdot \mathcal{E}_{\log,\Lambda}^{\rm cp}(N)+2 m^d N\left(\zeta_\Lambda'(0)-\zeta_{m \Lambda}'(0)\right).
  \end{align*}
 Then \eqref{cor1eq1} and \eqref{cor1eq2} follow from
\begin{align}
  \mathcal{E}_{s,m\Lambda}^{\rm cp}(m^d N)&=m^{-s}\mathcal{E}_{s,\Lambda}^{\rm cp}(m^d N), \qquad \zeta_{m\Lambda}(s)=m^{-s}\zeta_\Lambda(s),\label{eq25}\\
\mathcal{E}_{\log,m\Lambda}^{\rm cp}(m^d N)&=\mathcal{E}_{\log,\Lambda}^{\rm cp}(m^d N), \qquad \zeta_{m\Lambda}'(0)=\log m +\zeta_\Lambda'(0), \label{eq26}
\end{align}
where the first identity in (\ref{eq26}) is obtained from the first identity in (\ref{eq25}) using \eqref{Flog2} while the second identity in \eqref{eq26} follows by differentiating the second identity  in (\ref{eq25}) and evaluating at $s=0$.
\end{proof}


We are now ready to prove our main result.
\begin{proof}[Proof of Theorem \ref{thm2}]
By (\ref{eq27}) and (\ref{eq28}) it suffices to show
\begin{align*}
  &\overline{g}_{s,d}(\Lambda)=\underline{g}_{s,d}(\Lambda)=C_{s,d},\\
  &\overline{g}_{\log,d}(\Lambda)=\underline{g}_{\log,d}(\Lambda)=C_{\log,d}-2\zeta'_{\Lambda}(0).
\end{align*}
Fix some positive integer $N_0$. For any $N>N_0$ there exists $m\in\mathbb{N}$ such that $(m-1)^dN_0\leq N<m^d N_0$, using Corollary \ref{cor1} and the fact that $\{\frac{\mathcal{E}_{s,\Lambda}^{\rm cp}(N)}{N(N-1)}\}_{N=2}^\infty$ is an increasing sequence we obtain
\begin{align*}
  \frac{\mathcal{E}_{s,\Lambda}^{\rm cp}(N)}{N^{1+\frac{s}{d}}}&=\frac{\mathcal{E}_{s,\Lambda}^{\rm cp}(N)}{N(N-1)}\cdot\frac{N-1}{N^{\frac{s}{d}}}\leq \frac{\mathcal{E}_{s,\Lambda}^{\rm cp}(m^d N_0)}{m^d N_0(m^d N_0-1)}\cdot \frac{N-1}{N^{\frac{s}{d}}}\\
  &=\frac{\mathcal{E}_{s,\Lambda}^{\rm cp}(m^d N_0)}{(m^d N_0)^{1+\frac{s}{d}}}\cdot \frac{(m^d N_0)^{\frac{s}{d}}}{(m^d N_0-1)}\cdot \frac{N-1}{N^{\frac{s}{d}}}\\
  &\leq \left(\frac{\mathcal{E}_{s,\Lambda}^{\rm cp}(N_0)}{N_0^{1+\frac{s}{d}}}+\frac{(1-m^{-s})\zeta_\Lambda(s)}{N_0^\frac{s}{d}}\right)\cdot \frac{(m^d N_0)^{\frac{s}{d}}}{N^{\frac{s}{d}}}\cdot \frac{N-1}{(m^d N_0-1)}.
\end{align*}
Similarly
\begin{align*}
  \frac{\mathcal{E}_{\log,\Lambda}^{\rm cp}(N)+\frac{2}{d}N\log N}{N}&=\frac{\mathcal{E}_{\log,\Lambda}^{\rm cp}(N)}{N(N-1)}\cdot(N-1)+\frac{2}{d}\log N\\
&\leq \frac{\mathcal{E}_{\log,\Lambda}^{\rm cp}(m^d N_0)}{m^d N_0(m^d N_0-1)}\cdot(N-1)+\frac{2}{d}\log (m^d N_0)\\
&=\left(\frac{\mathcal{E}_{\log,\Lambda}^{\rm cp}(m^d N_0)+\frac{2}{d}m^d N_0 \log(m^d N_0)}{m^d N_0}-\frac{2}{d}\log(m^d N_0)\right)\cdot \frac{N-1}{m^d N_0-1}+\frac{2}{d}\log(m^d N_0)\\
&\leq \left(\frac{\mathcal{E}_{\log,\Lambda}^{\rm cp}(N_0)+\frac{2}{d}N_0\log N_0}{N_0}-\frac{2}{d}\log(m^d N_0)\right)\cdot \frac{N-1}{m^d N_0-1}+\frac{2}{d}\log(m^d N_0).
\end{align*}
Letting $\bnri$ yields
\begin{align*}
  \overline{g}_{s,d}(\Lambda)&=\limsup_{N\rightarrow \infty} \frac{\mathcal{E}_{s,\Lambda}^{\rm cp}(N)}{N^{1+s/d}}\leq \left(\frac{\mathcal{E}_{s,\Lambda}^{\rm cp}(N_0)}{N_0^{1+\frac{s}{d}}}+\frac{\zeta_\Lambda(s)}{N_0^\frac{s}{d}}\right),\\
\overline{g}_{\log,d}(\Lambda)&=\limsup_{N\rightarrow \infty} \frac{\mathcal{E}_{\log,\Lambda}^{\rm cp}(N)+\frac{2}{d}N\log N}{N}\leq \frac{\mathcal{E}_{\log,\Lambda}^{\rm cp}(N_0)+\frac{2}{d}N_0\log N_0}{N_0}.
\end{align*}
Letting $N_0\rightarrow\infty$ through an appropriate subsequence yields
\begin{align*}
\overline{g}_{s,d}(\Lambda)&\leq \liminf_{N_0\rightarrow \infty} \frac{\mathcal{E}_{s,\Lambda}^{\rm cp}(N_0)}{N_0^{1+\frac{s}{d}}}=\underline{g}_{s,d}(\Lambda),\\
\overline{g}_{\log,d}(\Lambda)&\leq \liminf_{N_0\rightarrow \infty} \frac{\mathcal{E}_{\log,\Lambda}^{\rm cp}(N_0)+\frac{2}{d}N_0\log N_0}{N_0}=\underline{g}_{\log,d}(\Lambda).
\end{align*}
Therefore $\overline{g}_{s,d}(\Lambda)=\underline{g}_{s,d}(\Lambda)=:C_{s,d}(\Lambda)$ and $\overline{g}_{\log,d}(\Lambda)=\underline{g}_{\log,d}(\Lambda)=:C_{\log,d}(\Lambda)$.

To show $C_{s,d}(\Lambda)$ is independent of $\Lambda$, let $\Lambda_1=A_1\mathbb{Z}^d$ and $\Lambda_2=A_2\mathbb{Z}^d$ be any two lattices with co-volume $1$. Then $\Lambda_2=Q\Lambda_1$ where $Q=A_2A_1^{-1}$. We can use rational matrices to approximate $Q$, namely, there exists a sequence $Q_m\in\frac{1}{m}GL(d;\mathbb{Z})$ such that $Q_m\rightarrow Q$.

For any lattice $\Lambda$, $mQ_m\Lambda=(m Q_m)\Lambda$ is a sublattice of $\Lambda$ since $mQ_m\in GL(d;\mathbb{Z})$.
Applying Lemma \ref{lem1} to $mQ_m\Lambda$ and $\Lambda$ we get
\begin{align*}
   \mathcal{E}_{s,mQ_m\Lambda}^{\rm cp}(Nm^d|\det Q_m|)\leq m^d|\det Q_m|\mathcal{E}_{s,\Lambda}^{\rm cp}( N)+Nm^d|\det Q_m|\left(\zeta_\Lambda(s)-\zeta_{mQ_m\Lambda}(s)\right).
\end{align*}
Now if we let $\Lambda=Q_m^{-1}\Lambda_2$ we get
\begin{align*}
  \mathcal{E}_{s,m\Lambda_2}^{\rm cp}(Nm^d|\det Q_m|)\leq m^d |\det Q_m|\mathcal{E}_{s,Q_m^{-1}\Lambda_2}^{\rm cp}( N)+Nm^d |\det Q_m|\left(\zeta_{Q_m^{-1}\Lambda_2}(s)-\zeta_{m\Lambda_2}(s)\right).
\end{align*}
Using relation (\ref{eq25}) again implies
\begin{align*}
    &m^{-s}\mathcal{E}_{s,\Lambda_2}^{\rm cp}(Nm^d|\det Q_m|)\leq m^d |\det Q_m|\mathcal{E}_{s,Q_m^{-1}\Lambda_2}^{\rm cp}( N)+Nm^d |\det Q_m|\left(\zeta_{Q_m^{-1}\Lambda_2}(s)-m^{-s}\zeta_{\Lambda_2}(s)\right),
\end{align*}
which can be rewritten as
\begin{align*}
    &\frac{\mathcal{E}_{s,\Lambda_2}^{\rm cp}(Nm^d |\det Q_m|)}{(Nm^d |\det Q_m|)^{1+\frac{s}{d}}}\leq \frac{\mathcal{E}_{s,Q_m^{-1}\Lambda_2}^{\rm cp}( N)}{N^{1+\frac{s}{d}}|\det Q_m|^\frac{s}{d}}+\frac{\zeta_{Q_m^{-1}\Lambda_2}(s)-m^{-s}\zeta_{\Lambda_2}(s)}{N^\frac{s}{d}|\det Q_m|^\frac{s}{d}}.
\end{align*}
Letting $m\rightarrow \infty$ and using Corollary~\ref{CorMinEn} and Lemma \ref{lem17}, we obtain
\begin{align*}
  C_{s,d}(\Lambda_2)\leq \frac{\mathcal{E}_{s,Q^{-1}\Lambda_2}^{\rm cp}(N)}{N^{1+\frac{s}{d}}}+\frac{\zeta_{Q^{-1}\Lambda_2}(s)}{N^\frac{s}{d}}=\frac{\mathcal{E}_{s,\Lambda_1}^{\rm cp}(N)}{N^{1+\frac{s}{d}}}+\frac{\zeta_{\Lambda_1}(s)}{N^\frac{s}{d}}.
\end{align*}
Taking $N\rightarrow \infty$ implies
\begin{align*}
  C_{s,d}(\Lambda_2)\leq C_{s,d}(\Lambda_1).
\end{align*}
By the arbitrariness of $\Lambda_1$ and $\Lambda_2$ we must have $C_{s,d}(\Lambda)\equiv C_{s,d}$ which is independent of $\Lambda$.

For the logarithmic case, we apply Lemma \ref{lem1} to $mQ_m\Lambda$ and $\Lambda$ to deduce
\begin{align*}
   \mathcal{E}_{\log,mQ_m\Lambda}^{\rm cp}(Nm^d|\det Q_m|)\leq m^d|\det Q_m|\mathcal{E}_{\log,\Lambda}^{\rm cp}( N)+2Nm^d|\det Q_m|\left(\zeta'_\Lambda(0)-\zeta'_{mQ_m\Lambda}(0)\right).
\end{align*}
Now if we let $\Lambda=Q_m^{-1}\Lambda_2$ we have
\begin{align*}
  \mathcal{E}_{\log,m\Lambda_2}^{\rm cp}(Nm^d |\det Q_m|)\leq m^d |\det Q_m|\mathcal{E}_{\log,Q_m^{-1}\Lambda_2}^{\rm cp}( N)+2Nm^d |\det Q_m|\left(\zeta'_{Q_m^{-1}\Lambda_2}(0)-\zeta'_{m\Lambda_2}(0)\right).
\end{align*}
Using relation (\ref{eq25}) again implies
\begin{align*}
    &\mathcal{E}_{\log,\Lambda_2}^{\rm cp}(Nm^d |\det Q_m|)\leq m^d |\det Q_m|\mathcal{E}_{\log,Q_m^{-1}\Lambda_2}^{\rm cp}( N)+2Nm^d |\det Q_m|\left(\zeta'_{Q_m^{-1}\Lambda_2}(0)-\log m-\zeta'_{\Lambda_2}(0)\right),
\end{align*}
which can be rewritten as
\begin{align*}
    &\frac{\mathcal{E}_{\log,\Lambda_2}^{\rm cp}(Nm^d|\det Q_m|)}{Nm^d|\det Q_m|}\leq \frac{\mathcal{E}_{\log,Q_m^{-1}\Lambda_2}^{\rm cp}( N)}{N}+2\left(\zeta'_{Q_m^{-1}\Lambda_2}(0)-\log m-\zeta'_{\Lambda_2}(0)\right).
\end{align*}
Therefore,
\begin{align*}
    &\frac{\mathcal{E}_{\log,\Lambda_2}^{\rm cp}(Nm^d |\det Q_m|)+\frac{2}{d}Nm^d |\det Q_m|\log(Nm^d |\det Q_m|)}{Nm^d |\det Q_m|}\\ \leq& \frac{\mathcal{E}_{\log,Q_m^{-1}\Lambda_2}^{\rm cp}( N)+\frac{2}{d}N\log N}{N}+2\left(\zeta'_{Q_m^{-1}\Lambda_2}(0)-\zeta'_{\Lambda_2}(0)\right)+\frac{2}{d}\log|\det Q_m|.
\end{align*}
Now let $m\rightarrow \infty$ and using Corollary~\ref{CorMinEn}   and Lemma \ref{lem17}
we obtain \begin{align*}
  C_{\log,d}(\Lambda_2)\leq \frac{\mathcal{E}_{\log,Q^{-1}\Lambda_2}^{\rm cp}( N)+\frac{2}{d}N\log N}{N}+2\left(\zeta'_{Q^{-1}\Lambda_2}(0)-\zeta'_{\Lambda_2}(0)\right).
\end{align*}
Taking $N\rightarrow \infty$ implies
\begin{align*}
  C_{\log,d}(\Lambda_2)\leq C_{\log,d}(\Lambda_1)+2\left(\zeta'_{\Lambda_1}(0)-\zeta'_{\Lambda_2}(0)\right).
\end{align*}
By symmetry
\begin{align*}
  C_{\log,d}(\Lambda_1)\leq C_{\log,d}(\Lambda_2)+2\left(\zeta'_{\Lambda_2}(0)-\zeta_{\Lambda_1}'(0)\right).
\end{align*}
It follows that
\begin{align*}
C_{\log,d}(\Lambda_1)+2\zeta_{\Lambda_1}'(0)=C_{\log,d}(\Lambda_2)+2\zeta_{\Lambda_2}'(0).
\end{align*}
Hence, if we define $C_{\log,d}:=C_{\log,d}(\Lambda)+2\zeta_{\Lambda}'(0)$ for any lattice $\Lambda$ of co-volume $1$, then this quantity is in fact independent of the choice of lattice $\Lambda$, which is what we wanted to show.
\end{proof}

\begin{section}{Proof of Theorem~\ref{shellthm}}\label{sect5}

\newcommand{\mcv}{\Lambda}
\newcommand{\lri}{L\to \infty}
\newcommand{\rmp}{p}
Recall that throughout this section  $\Lambda=A\Z^d$ denotes a $d$-dimensional lattice in $\R^d$ with fundamental domain $\Omega=\Omega_{\Lambda}$ and co-volume 1.  The $j$th column of the matrix $A$ is denoted by $v_j$.
First we establish the following lemma that will be used in the proof of Theorem~\ref{shellthm}.

\begin{lemma}\label{weaksphere}
If $s\leq d$ and $M\in\N$, then the unique weak limit of the measures
\[
\mu_L:=\frac{\sum_{v\in\mcv:M<|v|\leq L}|v|^{-s}\delta_{\frac{v}{|v|}}}{\sum_{v\in\mcv:M<|v|\leq L}|v|^{-s}}
\]
as $\lri$ is volume measure on the sphere $\Sph^{d-1}\subseteq\R^d$.
\end{lemma}
\begin{proof}

First note that as $|v|\to \infty$
\begin{equation}\label{x+v}\begin{split}
|x+v|^{-s}&=|v|^{-s}\left(1+2\frac{x\cdot v}{|v|^2}+\frac{|x|^2}{|v| ^2}\right)^{-s/2}\\ &=
|v|^{-s}\left[1-s\left(\frac{x\cdot v}{|v|^2}+\left(\frac{1}{2}\right)\frac{|x|^2}{|v| ^2}\right) +s(s+2)\left(\frac{ x\cdot v }{|v|^2}\right)^2\right]+O(|v|^{-s-3}),
\end{split}\end{equation}
which  implies:
\begin{equation} \label{x+vx-v}
|x+v|^{-s}+|x-v|^{-s}-2|v|^{-s}= -s\frac{|x|^2}{|v| ^{s+2}}+s(s+2)\left( \frac{ (x\cdot v/|v|)^2 }{|v|^{s+2}}\right) +O(|v|^{-s-3}), \qquad (|v|\to \infty).
\end{equation}



Let $\lfloor\cdot\rfloor_{\Lambda}$ denote \textit{lattice greatest integer function}  defined by
\[
\lfloor x\rfloor_{\Lambda}=A\cdot\sup\left\{k\in\Z^d:x-Ak=\sum_{j=1}^da_jv_j,\, a_j\geq0\right\},
\]
where we take the supremum with respect to the dictionary order on $\Z^d$.  

Using \eqref{x+vx-v} and writing $x=\lfloor x\rfloor_{\Lambda}+\{x\}_{\Lambda}$, where $\{x\}_{\Lambda}\in\Omega$
we obtain
\begin{equation}\label{claim1} |x|^{-s}-\left|\lfloor x\rfloor_{\Lambda}\right|^{-s}=O(|x|^{-s-1}), \qquad |x|\rightarrow\infty.
\end{equation}

\medskip

We shall also need the following result.
\noindent\underline{Claim:} Let $m_d$ denote Lebesgue measure on $\R^d$ and $0<s\le d$. Then
\begin{equation}\label{claim2}
\lim_{\lri}\frac{\sum_{ M<|v|\leq L}|v|^{-s}}{\int_{M<|x|\leq L}|x|^{-s}\mathrm{d}m_d(x)}=|\Omega|^{-1},
\end{equation}
\vspace{2mm}
where the sum  in the numerator of the left hand side of \eqref{claim2} is over $v\in \Lambda$ such that $M<|v|\leq L$.

To prove this claim, we first consider the case $s<d-1$.  We write
\begin{align}\label{ballint}
|\Omega|^{-1}\int_{M<|x|\leq L}|x|^{-s}\mathrm{d}m_d(x)-\sum_{M<|v|\leq L}|v|^{-s}&=|\Omega|^{-1}\int_{M<|x|\leq L}|x|^{-s}-|\lfloor x\rfloor_{\Lambda}|^{-s}\mathrm{d}m_d(x)+\epsilon_{L},
\end{align}
where
\begin{align}\label{errdef}
\epsilon_L=|\Omega|^{-1}\int_{M<|x|\leq L}|\lfloor x\rfloor_{\Lambda}|^{-s}\mathrm{d}m_d(x)-\sum_{M<|v|\leq L}|v|^{-s},
\end{align}
is an error term.  Suppose $(\Omega+v)\cap\{x:|x|=L\}=\emptyset$.  Then the contribution to (\ref{errdef}) from $v$ is zero.  Therefore, to estimate this error term, we need only consider vectors $v$ for which $(\Omega+v)\cap\{x:|x|=L\}\neq\emptyset$ in the sum and vectors $x$ for which $\lfloor x\rfloor_{\Lambda}=v$ and $(\Omega+v)\cap\{x:|x|=L\}\neq\emptyset$ in the integral.  This implies that $\epsilon_{L}=O(L^{d-1-s})$ as $\lri$.  Furthermore, \eqref{claim1} implies that the integral on the right-hand side of (\ref{ballint}) is $O(L^{d-1-s})$ as $\lri$.  Therefore, since $\sum_{M<|v|\leq L}|v|^{-s}$ grows like $L^{d-s}$ as $\lri$, we have proven \eqref{claim2} in the case $s<d-1$.

In the case $s=d-1$, the same reasoning shows that the integral on the right-hand side of (\ref{ballint}) is $O(\log(L))$ as $\lri$ and $\epsilon_L$ is bounded as $\lri$, while the sum $\sum_{M<|v|\leq L}|v|^{-s}$ grows like $L$ as $\lri$.  Therefore, \eqref{claim2} holds in this case as well.

If $d-1<s<d$, then the integral on the right-hand side of (\ref{ballint}) is $O(1)$ as $\lri$ and $\epsilon_L$ is $O(L^{d-1-s})$ as $\lri$, while the sum $\sum_{M<|v|\leq L}|v|^{-s}$ grows like $L^{d-s}$ as $\lri$ showing  \eqref{claim2} is true in this case as well.

Finally, if $s=d$, then the integral on the right-hand side of (\ref{ballint}) is $O(1)$ as $\lri$ and $\epsilon_L$ is $O(L^{-1})$ as $\lri$, while the sum $\sum_{M<|v|\leq L}|v|^{-s}$ grows like $\log(L)$ as $\lri$.  Therefore, \eqref{claim2} holds  for all $s\in(0,d]$.

\vspace{2mm}

Now we are ready to finish the proof of the lemma.  Let $f$ be any continuous function on $\Sph^{d-1}$ and consider
\begin{align}
\nonumber\lim_{\lri}\int_{\Sph^{d-1}}f(x)\mathrm{d}\mu_L(x)&=\lim_{\lri}\frac{\sum_{M<|v|\leq L}f\left(\frac{v}{|v|}\right)|v|^{-s}}{\sum_{M<|v|\leq L}|v|^{-s}}\\
\label{denom1}&=\left(1+o(1)\right)\lim_{\lri}\frac{|\Omega|\sum_{M<|v|\leq L}f\left(\frac{v}{|v|}\right)|v|^{-s}}{\int_{M<|x|\leq L}|x|^{-s}\mathrm{d}m_d(x)},
\end{align}
where we used \eqref{claim2}.  Notice that the denominator of this last expression grows like $L^{d-s}$ (or $\log(L)$ in the case $s=d$) as $\lri$.  Now consider the expression
\begin{align}\label{fadjust}
\nonumber&|\Omega|\sum_{M<|v|\leq L}f\left(\frac{v}{|v|}\right)|v|^{-s}-\int_{M<|x|\leq L}f\left(\frac{x}{|x|}\right)|x|^{-s}\mathrm{d}m_d(x)\\
&\qquad =\int_{M<|x|\leq L}f\left(\frac{\lfloor x\rfloor_{\Lambda}}{|\lfloor x\rfloor_{\Lambda}|}\right)|\lfloor x\rfloor_{\Lambda}|^{-s}\mathrm{d}m_d(x)-\int_{M<|x|\leq L}f\left(\frac{x}{|x|}\right)|x|^{-s}\mathrm{d}m_d(x)+\epsilon'_L,
\end{align}
where $\epsilon'_L$ is an error term.  An argument similar to the proof of the claim shows that $\epsilon'_L$ is negligibly small compared to the denominator in (\ref{denom1}) as $\lri$, so we may ignore this term when calculating the limit \eqref{errdef}.  We may then write
\begin{align}\label{denom2}
\nonumber&\int_{M<|x|\leq L}f\left(\frac{\lfloor x\rfloor_{\Lambda}}{|\lfloor x\rfloor_{\Lambda}|}\right)|\lfloor x\rfloor_{\Lambda}|^{-s}\mathrm{d}m_d(x)-\int_{M<|x|\leq L}f\left(\frac{x}{|x|}\right)|x|^{-s}\mathrm{d}m_d(x)\\
&\qquad\qquad\qquad=\int_{M<|x|\leq L}\left(f\left(\frac{\lfloor x\rfloor_{\Lambda}}{|\lfloor x\rfloor_{\Lambda}|}\right)-f\left(\frac{x}{|x|}\right)\right)|\lfloor x\rfloor_{\Lambda}|^{-s}\mathrm{d}m_d(x)\\
\nonumber&\qquad\qquad\qquad\qquad\qquad+\int_{M<|x|\leq L}f\left(\frac{x}{|x|}\right)\left(|\lfloor x\rfloor_{\Lambda}|^{-s}-|x|^{-s}\right)\mathrm{d}m_d(x).
\end{align}
Let us examine the first expression on the right-hand side of (\ref{denom2}).  Since the denominator on the right-hand side of (\ref{denom1}) tends to infinity as $\lri$, we may replace $M$ by $M_L$ tending to infinity very slowly as $\lri$ when calculating the first integral on the right-hand side of (\ref{denom2}).  In so doing, we introduce an error term that is negligible compared to the denominator on the right-hand side of (\ref{denom1}) as $\lri$.  Since $M_L\rightarrow\infty$ as $\lri$, the uniform continuity of $f$ implies that
\[
\lim_{\lri}\left|\frac{\int_{M<|x|\leq L}\left(f\left(\frac{\lfloor x\rfloor_{\Lambda}}{|\lfloor x\rfloor_{\Lambda}|}\right)-f\left(\frac{x}{|x|}\right)\right)|\lfloor x\rfloor_{\Lambda}|^{-s}dm_d(x)}{\int_{M<|x|\leq L}|x|^{-s}\mathrm{d}m_d(x)}\right|=0.
\]
Furthermore, \eqref{claim1} implies that the second integral on the right-hand side of (\ref{denom2}) is $O(L^{d-1-s})$ (or $O(\log(L))$ in the case $s=d-1$, or $O(1)$ if $s\in(d-1,d]$) as $\lri$.  We conclude that the limit on the right-hand side of (\ref{denom1}) is the same as
\[
\lim_{\lri}\frac{\int_{M<|x|\leq L}f\left(\frac{x}{|x|}\right)|x|^{-s}\mathrm{d}m_d(x)}{\int_{M<|x|\leq L}|x|^{-s}\mathrm{d}m_d(x)},
\]
which is clearly invariant under any symmetry of the sphere $\Sph^{d-1}$, so it must be equal to
\[
\int_{\Sph^{d-1}}f(z)\mathrm{d}\sigma_{d-1}(z)
\]
as desired.
\end{proof}

Now we are ready to prove Theorem \ref{shellthm}.

\begin{proof}[Proof of Theorem \ref{shellthm}]
Let $x\not\in \Lambda$ be fixed.
\begin{align*}
&\frac{|x|^{-s} +\frac{1}{2}\sum_{v\in\mcv_L\setminus\{0\}}\left(|x+v|^{-s}+|x-v|^{-s}-2|v|^{-s}\right)}{\sum_{v\in\mcv_L\setminus\{0\}}|v|^{-(s+2)}} \\
& \qquad= \frac{\frac{1}{2}\sum_{v\in\mcv_L\setminus\{0\}}-s\frac{|x|^2}{|v| ^{s+2}}+s(s+2)\left( \frac{ (x\cdot v/|v|)^2 }{|v|^{s+2}}\right) }{\sum_{v\in\mcv_L\setminus\{0\}}|v|^{-(s+2)}}+O(|v|^{-1}),
\end{align*}
as $|v|\to \infty$ where the implied constants depend only on $|x|$ and the distance from $x$ to $\Lambda$ and so the above holds uniformly  on compact subsets  of $ \R^d\setminus\Lambda$. Therefore, we have
\begin{equation*}
\begin{split}
\lim_{\lri}&\frac{|x|^{-s} +\frac{1}{2}\sum_{v\in\mcv_L\setminus\{0\}}\left(|x+v|^{-s}+|x-v|^{-s}-2|v|^{-s}\right)}{\sum_{v\in\mcv_L\setminus\{0\}}|v|^{-(s+2)}} \\ &=\lim_{\lri}\int_{\Sph^{d-1}}\left(-s|x|^2+s(s+2)\langle x,z\rangle^2\right)\mathrm{d}\mu_L(z)
\\ &=\int_{\Sph^{d-1}}\left(-s|x|^2+s(s+2)\langle x,z\rangle^2\right) \mathrm{d}\sigma_{d-1}(z)),
\\
&=s\left(\frac{s+2}{d}-1\right)|x|^2,
\end{split}
\end{equation*}
where we used  Lemma \ref{weaksphere}, but with $s$ replaced by $s+2$.  \end{proof}

\end{section}


\bibliographystyle{plain}

\bibliography{P2bibfile}
\end{document}